\documentclass{article}   	
               						
\usepackage {amsmath,url}
\usepackage {amssymb}
\usepackage {bm}
\usepackage {graphicx}
\usepackage {graphpap}
\usepackage {longtable}
\usepackage {multirow}
\usepackage {amsthm,afterpage}
\pdfoutput=1

\allowdisplaybreaks

\newcommand{\T}{\textsf{T}}
\theoremstyle{plain}
\newtheorem{theorem}{Theorem}
\newtheorem{lemma}[theorem]{Lemma}

\theoremstyle{definition}
\newtheorem{example}[theorem]{Example}
\theoremstyle{remark}
\newtheorem*{remark}{Remark} 
\usepackage {tikz}
\usetikzlibrary {decorations.markings}

\usepackage {varioref}
\renewcommand{\P}{{\rm P}}

\title{Spatial Parrondo games \\ with spatially dependent game $A$}
\author{Sung Chan Choi\thanks{Department of Mathematics, University of Utah, 155 S. 1400 E., Salt Lake City, UT 84112, USA. e-mail: choi@math.utah.edu}}
\date{}							

\begin{document}
\maketitle

\begin{abstract}
Parrondo games with spatial dependence were introduced by Toral (2001) and have been studied extensively.  In Toral's model, $N$ players are arranged in a circle.  The players play either game $A$ or game $B$. In game $A$, a randomly chosen player wins or loses one unit according to the toss of a fair coin.  In game $B$, which depends on parameters $p_0,p_1,p_2\in[0,1]$, a randomly chosen player, player $x$ say, wins or loses one unit according to the toss of a $p_m$-coin, where $m\in\{0,1,2\}$ is the number of nearest neighbors of player $x$ who won their most recent game.  In this paper, we replace game $A$ by a spatially dependent game, which we call game $A'$, introduced by Xie et al.~(2011).  In game $A'$, two nearest neighbors are chosen at random, and one pays one unit to the other based on the toss of a fair coin.  Noting that game $A'$ is fair, we say that the \textit{Parrondo effect} occurs if game $B$ is losing or fair and game $C'$, determined by a random or periodic sequence of games $A'$ and $B$, is winning.  We investigate numerically the region in which the Parrondo effect appears.  We give sufficient conditions for the mean profit in game $C'$ to converge as $N\to\infty$.  Finally, we compare the Parrondo region in the model of Xie et al.\ with that in the model of Toral.
\end{abstract}

\section{Introduction}

The Parrondo effect, in which there is a reversal in direction in some system parameter when two similar dynamics are combined, is the result of an underlying nonlinearity.  It was first described by Spanish physicist J. M. R. Parrondo in 1996 in the context of games of chance: He showed that it is possible to combine two fair or losing games, $A$ and $B$, to produce a winning one, $C$.  Here $C$ is the game obtained by playing games $A$ and $B$ in a random or periodic sequence.  His motivation was to provide a discrete (in time and space) version of the so-called flashing Brownian ratchet of Ajdari and Prost \cite{AP92}.  Other versions of Parrondo's games followed, including Toral's \cite{T01} spatially dependent games.  These games were modified by Xie et al.~\cite{XC11}, and it is the goal of this paper to explore the latter games in greater depth than was done by Ethier and Lee \cite{EL15}.

\subsection{Toral's spatially dependent games}

Toral \cite{T01} introduced what he called \textit{cooperative} Parrondo games with spatial dependence.  (We prefer the term \textit{spatially dependent} Parrondo games so as to avoid conflict with the field of cooperative game theory.)  The games depend on an integer parameter $N\ge3$, the number of players, and four probability parameters, $p_0,p_1,p_2,p_3$.  (This is a slight generalization of the model described in the abstract.) The players are arranged in a circle and labeled from 1 to $N$ (so that players 1 and $N$ are adjacent).  At each turn, a player is chosen at random to play.  Suppose player $x$ is chosen.  In game $A$, he tosses a fair coin.  In game $B$, he tosses a $p_m$-coin (i.e., a coin whose probability of heads is $p_m$), where $m\in\{0,1,2,3\}$ depends on the winning or losing status of his two nearest neighbors.  A player's status as winner (1) or loser (0) is decided by the result of his most recent game.  Specifically,
$$
m=\begin{cases}0&\text{if $x-1$ and $x+1$ are both losers,}\\
               1&\text{if $x-1$ is a loser and $x+1$ is a winner,}\\
               2&\text{if $x-1$ is a winner and $x+1$ is a loser,}\\
               3&\text{if $x-1$ and $x+1$ are both winners,}\end{cases}
$$
where $N+1:=1$ and $0:=N$ because of the circular arrangement of players.  Player $x$ wins one unit with heads and loses one unit with tails.  Replacing $(p_0.p_1,p_2,p_3)$ by $(p_0,p_1,p_1,p_2)$ gives the 3-parameter model described in the abstract.

These games have been studied in detail in a series of papers by Ethier and Lee \cite{EL12a,EL12b,EL13a,EL13b}.  For example, with Toral's \cite{T01} choice of parameters, namely $(p_0,p_1,p_2,p_3)=(1,0.16,0.16,0.7)$, one can compute the asymptotic profit per turn to the set of $N$ players, for $3\le N\le 19$.  For $N=5,6$ and $9\le N\le19$, the Parrondo effect (where game $A$ is fair, game $B$ is fair or losing, and the random mixture, game $C:=\frac12 A+\frac12 B$, is winning) is present.  In the cited papers, a strong law of large numbers and a central limit theorem are obtained.  In particular, the asymptotic cumulative profits per turn exist and are the means in the SLLN.  Further, it seems clear that these means converges as $N\to\infty$.  This has been proved under certain conditions (see Ethier and Lee \cite{EL13a}).

\subsection{The spatially dependent games of Xie et al.}

Notice that Toral's \cite{T01} game $A$ is not spatially dependent (i.e., the rules of the game do not depend on the spatial structure of the players). Xie et al.~\cite{XC11} proposed a modification of game $A$ that \textit{is} spatially dependent as well as being a fair game.  To distinguish, we call that game $A'$.  As before, the games depend on an integer parameter $N\ge3$, the number of players, and four probability parameters, $p_0,p_1,p_2,p_3$.  The players are arranged in a circle and labeled from $1$ to $N$ (so that players $1$ and $N$ are adjacent).  At each turn, a player is chosen at random to play.  Suppose player $x$ is chosen.  In game $A'$, he chooses one of his two nearest neighbors at random and competes with that neighbor by tossing a fair coin.  The results is a transfer of one unit from one of the players to the other, hence the wealth of the set of $N$ players is unchanged.  Game $B$ is as before.  Player $x$ wins one unit with heads and loses one unit with tails. 

These games were studied by Xie et al.~\cite{XC11}, Li et al.~\cite{LY14}, and Ethier and Lee \cite{EL15}.  Only the random mixture case was treated, and convergence of the means has not yet been addressed.  Our aim in this paper is to fill in these gaps in the literature.  Further, we want to understand this model as well as Toral's model is understood.

We begin by establishing a strong law of large numbers and a central limit theorem, especially in the periodic pattern case.  We compute various means numerically and use computer graphics to visualize the Parrondo region.  Then we address the issue of convergence of means, which involves certain interacting particle systems.  We then establish the convergence, both in the random mixture setting and in the periodic pattern setting, on a large subset of the parameter space.

\section{SLLN/CLT for the games of Xie et al.}
       
In this section, we restate the strong law of large numbers (SLLN) and the central limit theorem (CLT) of Ethier and Lee \cite{EL09}, and we apply them to the Parrondo games of Xie at al.~\cite{XC11}.

\subsection{SLLN and CLT}\label{SLLN-CLT-EL}

Ethier and Lee \cite{EL09} proved an SLLN and a CLT for the Parrondo player's sequence of profits, motivated by the random mixture $C:=\gamma A+(1-\gamma) B$.  A subsequent version in the same paper treats the case of periodic patterns.

Consider an irreducible aperiodic Markov chain $\{X_n\}_{n\ge0}$ with finite state space $\Sigma_0$.  It evolves according to the one-step transition matrix ${\bm P}=(P_{ij})_{i,j\in\Sigma_0}$.  Let us denote its unique stationary distribution by the row vector ${\bm \pi}=(\pi_i)_{i\in \Sigma_0}$.  Let $w:\Sigma_0\times\Sigma_0\mapsto {\bf R}$ be an arbitrary function, which we write as a matrix ${\bm W}=(w(i,j))_{i,j\in\Sigma_0}$ and refer to as the \textit{payoff matrix}. Define the sequences $\{\xi_n\}_{n\ge1}$ and $\{S_n\}_{n\ge1}$ by
\begin{equation*}
\xi_n:=w(X_{n-1},X_n),\qquad n\ge1,
\end{equation*}
and
\begin{equation*}
S_n:=\xi_1+\cdots+\xi_n,\qquad n\ge1.
\end{equation*}
Let ${\bm \Pi}$ denote the square matrix each of whose rows is ${\bm \pi}$, and let ${\bm Z}:=({\bm I}-({\bm P}-{\bm \Pi}))^{-1}$ denote the \textit{fundamental matrix}.  Denote by $\dot{\bm P}$ and $\ddot{\bm P}$ the Hadamard (entrywise) products $\bm P\circ\bm W$ and $\bm P\circ\bm W\circ\bm W$ (so $\dot{P}_{ij}:=P_{ij}w(i,j)$ and $\ddot{P}_{ij}:=P_{ij}w(i,j)^2$).  Let $\bm 1:=(1,1,\ldots,1)^\T$ and define
\begin{equation*}
\mu:=\bm\pi\dot{\bm P}\bm 1\quad{\rm and}\quad\sigma^2:=\bm\pi\ddot{\bm P}\bm 1
-(\bm\pi\dot{\bm P}\bm 1)^2+2\bm\pi\dot{\bm P}(\bm Z-\bm\Pi)\dot{\bm P}\bm 1.
\end{equation*}

\begin{theorem}[Ethier and Lee \cite{EL09}]\label{SLLN}
Under the above assumptions, and with the distribution of $X_0$ arbitrary,  
$$
\frac{S_n}{n}\to \mu\;\;{\rm a.s.}
$$
and, if $\sigma^2>0$, 
$$
\frac{S_n-n\mu}{\sqrt{n\sigma^2}}\to_d N(0,1).
$$
If $\mu=0$ and $\sigma^2>0$, then $-\infty=\liminf_{n\to\infty}S_n<\limsup_{n\to\infty}S_n=\infty$ \emph{a.s.}
\end{theorem}

\begin{example}
To illustrate this theorem, let us consider the original capital-dependent Parrondo games (without a bias parameter).  These are single-player games.  In game $A$, the player tosses a fair coin.  In game $B$, the player tosses a 1/10-coin if capital is divisible by 3 and and a 3/4-coin otherwise.  In either case, the player wins one unit with heads and loses one unit with tails.  The underlying Markov chain corresponding to game $B$ has state space $\Sigma_0:=\{0,1,2\}$ and one-step transition matrix
$$
\bm P_B:=\begin{pmatrix}0&1/10&9/10\\1/4&0&3/4\\3/4&1/4&0\end{pmatrix}.
$$
Its unique stationary distribution is $\bm\pi_B=(1/13)(5,2,6)$.  The payoff matrix has the form
$$
\bm W:=\begin{pmatrix}0&1&-1\\-1&0&1\\1&-1&0\end{pmatrix}.
$$
We find that
$$
\mu_B=\bm\pi_B\dot{\bm P}_B\bm1=0.
$$
The underlying Markov chain corresponding to game $A$ has the same state space and one-step transition matrix
$$
\bm P_A:=\begin{pmatrix}0&1/2&1/2\\1/2&0&1/2\\1/2&1/2&0\end{pmatrix}
$$
with unique stationary distribution $\bm\pi_A=(1/3)(1,1,1)$.  The payoff matrix is the same, and we find that
$$
\mu_A=\bm\pi_A\dot{\bm P}_A\bm1=0,
$$
a result that is obvious without calculation.
Finally, the underlying Markov chain corresponding to game $C:=\frac12 A+\frac12 B$ has the same state space  and one-step transition matrix
$$
\bm P_C:=\frac12(\bm P_A+\bm P_B)=\begin{pmatrix}0&3/10&7/10\\3/8&0&5/8\\5/8&3/8&0\end{pmatrix}
$$
with unique stationary distribution $\bm\pi_C=(1/709)(245,180,284)$.  The payoff matrix is the same, and we find that
$$
\mu_C=\bm\pi_C\dot{\bm P}_C\bm1=\frac{18}{709}\approx0.0253879.
$$
This is perhaps the best-known example of Parrondo's paradox, and the SLLN justifies the conclusion:  Two fair games combine to win.

We can also derive a CLT, which requires the fundamental matrix
$$
\bm Z_B:=(\bm I-(\bm P_B-\bm \Pi_B))^{-1}=\frac{1}{2197}\begin{pmatrix}1725&-38&510\\-95&1938&354\\425&118&1654\end{pmatrix}.
$$
We find that
$$
\sigma^2_B=\bm\pi_B\ddot{\bm P}_B\bm 1
-(\bm\pi_B\dot{\bm P}_B\bm 1)^2+2\bm\pi_B\dot{\bm P}_B(\bm Z_B-\bm\Pi_B)\dot{\bm P}_B\bm 1=\bigg(\frac{9}{13}\bigg)^2\approx0.479290.
$$
Similarly,
$$
\bm Z_A:=(\bm I-(\bm P_A-\bm \Pi_A))^{-1}=\frac{1}{9}\begin{pmatrix}7&1&1\\1&7&1\\1&1&7\end{pmatrix},
$$
hence
$$
\sigma^2_A=\bm\pi_A\ddot{\bm P}_A\bm 1
-(\bm\pi_A\dot{\bm P}_A\bm 1)^2+2\bm\pi_A\dot{\bm P}_A(\bm Z_A-\bm\Pi_A)\dot{\bm P}_A\bm 1=1,
$$
as is obvious without the formula.  Finally, 
$$
\bm Z_C:=(\bm I-(\bm P_C-\bm \Pi_C))^{-1}=\frac{1}{502681}\begin{pmatrix}392265&22884&87532\\23585&408580&70516\\80305&39900&382476\end{pmatrix},
$$
and we conclude that
$$
\sigma^2_C=\bm\pi_C\ddot{\bm P}_C\bm 1
-(\bm\pi_C\dot{\bm P}_C\bm 1)^2+2\bm\pi_C\dot{\bm P}_C(\bm Z_C-\bm\Pi_C)\dot{\bm P}_C\bm 1=\frac{311313105}{356400829}\approx0.873492.
$$
In each case we have a CLT.
\end{example}

Next we turn to another SLLN and CLT of Ethier and Lee \cite{EL09}, this one motivated by the case of periodic patterns.

Let $\bm P_A$ and $\bm P_B$ be one-step transition matrices for Markov chains in a finite state space $\Sigma_0$.  Fix integers $r,s\ge1$.  Assume that $\bm P:=\bm P_A^r\bm P_B^s$, as well as all cyclic permutations of $\bm P_A^r\bm P_B^s$, are ergodic, and let the row vector $\bm\pi$ be the unique stationary distribution of $\bm P$.  Let $\bm\Pi$ be the square matrix each of whose rows is equal to $\bm\pi$, and let ${\bm Z}:=({\bm I}-({\bm P}-{\bm \Pi}))^{-1}$ be the fundamental matrix of $\bm P$.  Given a real-valued function $w$ on $\Sigma_0\times\Sigma_0$,  define the payoff matrix $\bm W:=(w(i,j))_{i,j\in\Sigma_0}$.  Define $\dot{\bm P}_A:=\bm P_A\circ\bm W$, $\dot{\bm P}_B:=\bm P_B\circ\bm W$, $\ddot{\bm P}_A:=\bm P_A\circ\bm W\circ\bm W$, $\ddot{\bm P}_B:=\bm P_B\circ\bm W\circ\bm W$, where $\circ$ denotes the Hadamard (entrywise) product.  Let
\begin{equation*}
\mu_{[r,s]}:=\frac{1}{r+s}\bigg[\sum_{u=0}^{r-1} {\bm\pi}{\bm P}_A^u\dot{\bm P}_A\bm1+\sum_{v=0}^{s-1}
{\bm\pi}{\bm P}_A^r{\bm P}_B^v\dot{\bm P}_B\bm1\bigg],
\end{equation*}
and
\begin{align*}
\sigma_{[r,s]}^2
&=\frac{1}{r+s}\bigg[\sum_{u=0}^{r-1}[\bm \pi\bm P_A^u\ddot{\bm P}_A\bm1-(\bm\pi\bm P_A^u\dot{\bm P}_A\bm1)^2]\nonumber\\
&\qquad\quad+\sum_{v=0}^{s-1}[\bm\pi\bm P_A^r\bm P_B^v\ddot{\bm P}_B\bm1-(\bm\pi\bm P_A^r\bm P_B^v\dot{\bm P}_B\bm1)^2]\nonumber\\
&\qquad\quad{}+2\sum_{0\le u<v\le r-1}\bm\pi\bm P_A^u\dot{\bm P}_A(\bm P_A^{v-u-1}-\bm \Pi\bm P_A^v)\dot{\bm P}_A\bm1\nonumber\\
&\qquad\quad{}+2\sum_{u=0}^{r-1}\sum_{v=0}^{s-1}\bm\pi\bm P_A^u\dot{\bm P}_A(\bm P_A^{r-u-1}-\bm\Pi\bm P_A^r)\bm P_B^v\dot{\bm P}_B\bm1\nonumber\\
&\qquad\quad{}+2\sum_{0\le u<v\le s-1}\bm\pi\bm P_A^r\bm P_B^u\dot{\bm P}_B(\bm P_B^{v-u-1}-\bm\Pi\bm P_A^r\bm P_B^v)\dot{\bm P}_B\bm1\nonumber\\
&\qquad\quad{}+2\bigg(\sum_{u=0}^{r-1}\sum_{v=0}^{r-1}\bm\pi\bm P_A^u\dot{\bm P}_A\bm P_A^{r-u-1}\bm P_B^s(\bm Z-\bm \Pi)\bm P_A^v\dot{\bm P}_A\bm1\nonumber\\
&\qquad\qquad\quad{}+\sum_{u=0}^{r-1}\sum_{v=0}^{s-1}\bm\pi\bm P_A^u\dot{\bm P}_A\bm P_A^{r-u-1}\bm P_B^s(\bm Z-\bm\Pi)\bm P_A^r\bm P_B^v\dot{\bm P}_B\bm1\nonumber\\
&\qquad\qquad\quad{}+\sum_{u=0}^{s-1}\sum_{v=0}^{r-1}\bm\pi\bm P_A^r\bm P_B^u\dot{\bm P}_B\bm P_B^{s-u-1}(\bm Z-\bm\Pi)\bm P_A^v\dot{\bm P}_A\bm1\nonumber\\
&\qquad\qquad\quad{}+\sum_{u=0}^{s-1}\sum_{v=0}^{s-1}\bm\pi\bm P_A^r\bm P_B^u\dot{\bm P}_B\bm P_B^{s-u-1}(\bm Z-\bm\Pi)\bm P_A^r\bm P_B^v\dot{\bm P}_B\bm1\bigg)\bigg],
\end{align*}
where $\bm1$ denotes a column vector of $1$s with entries indexed by $\Sigma_0$.  Let $\{X_n\}_{n\ge0}$ be a nonhomogeneous Markov chain in $\Sigma_0$ with one-step transition matrices $\bm P_A,\ldots,\bm P_A$ $(r\text{ times})$, $\bm P_B,\ldots,\bm P_B$ $(s\text{ times})$, $\bm P_A,\ldots,\bm P_A$ $(r\text{ times})$, $\bm P_B,\ldots,\bm P_B$ $(s\text{ times})$, and so on.  For each $n\ge1$, define $\xi_n:=w(X_{n-1},X_n)$ and $S_n:=\xi_1+\cdots+\xi_n$.  

\begin{theorem}[Ethier and Lee \cite{EL09}]\label{SLLN2}
Under the above assumptions, and with the distribution of $X_0$ arbitrary, 
$$
\frac{S_n}{n}\to \mu_{[r,s]}\;\;{\rm a.s.}
$$
and, if $\sigma_{[r,s]}^2>0$, then 
$$
\frac{S_n-n\mu_{[r,s]}}{\sqrt{n\sigma_{[r,s]}^2}}\to_d N(0,1) \text{ as } n\to\infty.
$$
\end{theorem}

\begin{example}
To illustrate this result, we consider the capital-dependent Parrondo games as above, and we take $r=s=2$.  Then
$$
\bm P=\bm P_A^2\bm P_B^2=\frac{1}{320}\begin{pmatrix}162&59&99\\151&58&111\\111&47&162\end{pmatrix}.
$$
Its unique stationary distribution is $\bm\pi=(1/6357)(2783, 1075, 2499)$, and the fundamental matrix is
$$
\bm Z=\frac{1}{525348837}\begin{pmatrix}569627023& 10027235& -54305421\\ 22416463& 532826915& -29894541\\ -58953137& -14383645& 598685619\end{pmatrix}.
$$
In this example, $\dot{\bm P}_A\bm1=\bm0$, $\ddot{\bm P}_A=\bm P_A$, and $\ddot{\bm P}_B=\bm P_B$, and this simplifies the mean and variance formulas considerably.  Specifically, we have
$$
\mu_{[2,2]}=\frac{1}{4}\bm\pi\bm P_A^2(\bm I+\bm P_B)\dot{\bm P}_B\bm1
$$
and 
\begin{align*}
\sigma_{[2,2]}^2&=\frac{1}{4}\big[2+2-(\bm\pi\bm P_A^2\dot{\bm P}_B\bm1)^2-(\bm\pi\bm P_A^2\bm P_B\dot{\bm P}_B\bm1)^2\\
&\quad+2\bm\pi\dot{\bm P}_A(\bm P_A-\bm\Pi\bm P_A^2)(\bm I+\bm P_B)\dot{\bm P}_B\bm1\\
&\quad+2\bm\pi\bm P_A\dot{\bm P}_A(\bm I-\bm\Pi\bm P_A^2)(\bm I+\bm P_B)\dot{\bm P}_B\bm1\\
&\quad+2\bm\pi\bm P_A^2\dot{\bm P}_B(\bm I-\bm\Pi\bm P_A^2\bm P_B)\dot{\bm P}_B\bm1\\
&\quad+2\bm\pi\dot{\bm P}_A\bm P_A\bm P_B^2(\bm Z-\bm\Pi)\bm P_A^2(\bm I+\bm P_B)\dot{\bm P}_B\bm1\\
&\quad+2\bm\pi\bm P_A\dot{\bm P}_A\bm P_B^2(\bm Z-\bm\Pi)\bm P_A^2(\bm I+\bm P_B)\dot{\bm P}_B\bm1\\
&\quad+2\bm\pi\bm P_A^2\dot{\bm P}_B\bm P_B(\bm Z-\bm\Pi)\bm P_A^2(\bm I+\bm P_B)\dot{\bm P}_B\bm1\\
&\quad+2\bm\pi\bm P_A^2\bm P_B\dot{\bm P}_B(\bm Z-\bm\Pi)\bm P_A^2(\bm I+\bm P_B)\dot{\bm P}_B\bm1\big].
\end{align*}
We conclude that
$$
\mu_{[2,2]}=\frac{4}{163}\approx0.0245399\quad\text{and}\quad\sigma_{[2,2]}^2=\frac{1923037543}{2195688729}\approx0.875824.
$$
These numbers are consistent with Ethier and Lee \cite{EL09}.
\end{example}

\subsection{Application to game $B$}\label{SLLN-gameB}

The Markov chain formalized by Mihailovi\'c and Rajkovi\'c \cite{MR03} keeps track of the status (loser or winner, 0 or 1) of each of the $N\ge3$ players of game $B$. Its state space is the product space 
$$
\{\eta=(\eta(1),\eta(2),\ldots,\eta(N)): \eta(x)\in\{0,1\}{\rm\ for\ }x=1,\ldots,N\}=\{0,1\}^N
$$
with $2^N$ states.  Let $m_x(\eta):=2\eta(x-1)+\eta(x+1)\in\{0,1,2,3\}$.  Of course $\eta(0):=\eta(N)$ and $\eta(N+1):=\eta(1)$ because of the circular arrangement of players.  Also, let $\eta_x$ be the element of $\{0,1\}^N$ equal to $\eta$ except at the $x$th coordinate.  For example, $\eta_1:=(1-\eta(1),\eta(2),\eta(3),\ldots,\eta(N))$.

The one-step transition matrix $\bm P_B$ for this Markov chain depends not only on $N$ but on four parameters, $p_0,p_1,p_2,p_3\in[0,1]$.  It has the form
\begin{equation}\label{PB1}
P_B(\eta,\eta_x):=\begin{cases}N^{-1}p_{m_x(\eta)}&\text{if $\eta(x)=0$,}\\N^{-1}q_{m_x(\eta)}&\text{if $\eta(x)=1$,}\end{cases}\qquad x=1,\ldots,N,\;\eta\in\{0,1\}^N,
\end{equation}
and 
\begin{equation}\label{PB2}
P_B(\eta,\eta):=N^{-1}\bigg(\sum_{x:\eta(x)=0}q_{m_x(\eta)}+\sum_{x:\eta(x)=1}p_{m_x(\eta)}\bigg),\qquad \eta\in\{0,1\}^N,
\end{equation}
where $q_m:=1-p_m$ for $m=0,1,2,3$ and empty sums are 0.  The Markov chain is irreducible and aperiodic if $0<p_m<1$ for $m=0,1,2,3$.  Under slightly weaker assumptions (see Ethier and Lee \cite{EL13a}), the Markov chain is ergodic, which suffices.  For example, if $p_0$ is arbitrary and $0<p_m<1$ for $m=1,2,3$, or if $0<p_m<1$ for $m=0,1,2$ and $p_3$ is arbitrary, then ergodicity holds.

It appears at first glance that the theorem does not apply in the context of game $B$ because the payoffs are not completely specified by the one-step transitions of the Markov chain.  Specifically, a transition from  a state $\eta$ to itself results whenever a loser loses or a winner wins, so the transition does not determine the payoff.  

Our original Markov chain has state space $\{0,1\}^N$ and its one-step transition matrix $\bm P_B$ is given by \eqref{PB1} and \eqref{PB2}.  Assuming it is ergodic, let $\bm\pi_B$ denote its unique stationary distribution.  The approach in Ethier and Lee \cite{EL12a} augments the state space, letting $\Sigma^*:=\{0,1\}^N\times \{1,2,\ldots,N\}$ and keeping track not only of the status of each player as described by $\eta\in\{0,1\}^N$ but also of the label of the next player to play, say $x$.  The new one-step transition matrix $\bm P_B^*$ can be determined, as can its unique stationary distribution $\bm\pi_B^*$, and the theorem applies.

However, there is a drawback to this approach, namely that it is not clear that the variance parameter $(\sigma^*)^2$ is the same as the original one, $\sigma^2$.  (It is easy to verify that $\mu^*=\mu$.)  Therefore, we take a different approach, namely the one used by Ethier and Lee \cite{EL17} in their study of two-dimensional spatial models.

Here a different augmentation of $\{0,1\}^N$ is more effective.  We let $\Sigma^\circ:=\{0,1\}^N\times \{-1,1\}$ and keep track not only of $\eta\in\{0,1\}^N$ but also of the profit from the last game played, say $s\in\{-1,1\}$.  The new one-step transition matrix $\bm P_B^\circ$ has the form, for every $(\eta,s)\in\Sigma^\circ$,
\begin{equation*}
P_B^\circ((\eta,s),(\eta_x,1)):=\begin{cases}N^{-1}p_{m_x(\eta)}&\text{if $\eta(x)=0$,}\\
0&\text{if $\eta(x)=1$,}\end{cases}
\end{equation*}
\begin{equation*}
P_B^\circ((\eta,s),(\eta_x,-1)):=\begin{cases}0&\text{if $\eta(x)=0$,}\\
N^{-1}q_{m_x(\eta)}&\text{if $\eta(x)=1$,}\end{cases}
\end{equation*}
for $x=1,\ldots,N$, and
\begin{equation*}
P_B^\circ((\eta,s),(\eta,1)):=N^{-1}\sum_{x:\eta(x)=1}p_{m_x(\eta)},
\end{equation*}
\begin{equation*}
P_B^\circ((\eta,s),(\eta,-1)):=N^{-1}\sum_{x:\eta(x)=0}q_{m_x(\eta)},
\end{equation*}
where $q_m:=1-p_m$ for $m=0,1,2,3,4$ and $m_x(\eta)=2\eta(x-1)+\eta(x+1)$.  There are two inaccessible states, $(\bm 0,1)$ and $(\bm 1,-1)$, but the Markov chain remains ergodic.  Let $\bm\pi_B^\circ$ denote the unique stationary distribution, which has entry 0 at each of the two inaccessible states.  The payoff function $w^\circ$ can now be defined by
\begin{equation*}
w^\circ((\eta,s),(\eta_x,t))=t\text{ if $\eta(x)=(1-t)/2$,}\qquad w^\circ((\eta,s),(\eta,t))=t
\end{equation*}
for all $(\eta,s)\in\Sigma^\circ$, $x=1,2,\ldots,N$, and $t\in\{-1,1\}$, and $w^\circ=0$ otherwise.  This allows us to define the matrix $\bm W^\circ$ and then $\dot{\bm P}_B^\circ:={\bm P}_B^\circ\circ\bm W^\circ$ and $\ddot{\bm P}_B^\circ:={\bm P}_B^\circ\circ\bm W^\circ\circ\bm W^\circ$, the Hadamard (or entrywise) products.  Theorem~\ref{SLLN} yields the following.

Let $0<p_m<1$ for $m=0,1,2$ or for $m=1,2,3$, so that the Markov chain with one-step transition matrix $\bm P_B^\circ$ is ergodic, and let the row vector $\bm\pi_B^\circ$ be its unique stationary distribution.  Define
\begin{equation*}
\mu_B^\circ=\bm\pi_B^\circ\dot{\bm P}_B^\circ\bm1,\qquad (\sigma_B^\circ)^2=\bm\pi_B^\circ\ddot{\bm P}_B^\circ\bm 1-(\bm\pi_B^\circ\dot{\bm P}_B^\circ\bm 1)^2+2\bm\pi_B^\circ\dot{\bm P}_B^\circ(\bm Z_B^\circ-\bm1\bm\pi_B^\circ)\dot{\bm P}_B^\circ\bm 1.
\end{equation*}
where $\bm1$ denotes a column vector of $1$s with entries indexed by $\Sigma_B^\circ$ and $\bm Z_B^\circ:=(\bm I-(\bm P_B^\circ-\bm1\bm\pi_B^\circ))^{-1}$ is the fundamental matrix.  (Notice that $\bm1\bm\pi_B^\circ$ is the square matrix each of whose rows is equal to $\bm\pi_B^\circ$.)  Let $\{X_n^\circ\}_{n\ge0}$ be a time-homogeneous Markov chain in $\Sigma^\circ$ with one-step transition matrix $\bm P_B^\circ$, and let the initial distribution be arbitrary.  For each $n\ge1$, define $\xi_n:=w^\circ(X_{n-1}^\circ,X_n^\circ)$ and $S_n:=\xi_1+\cdots+\xi_n$.

\begin{theorem}\label{SLLN/CLT-PBcirc}
Under the above assumptions, and with the initial distribution arbitrary,
$$
\lim_{n\to\infty}\frac{S_n}{n}=\mu_B^\circ\;\;\emph{a.s.} 
$$
and, if $(\sigma_B^\circ)^2>0$, then 
$$
\frac{S_n-n\mu_B^\circ}{\sqrt{n(\sigma_B^\circ)^2}}\to_d N(0,1)\text{ as }n\to\infty.
$$
\end{theorem}

We next show that there is a simpler expression for this mean and variance.  Let us define
\begin{equation*}
\mu_B:=\bm\pi_B\dot{\bm P}_B\bm1,\qquad\sigma_B^2:=\bm\pi_B\ddot{\bm P}_B\bm 1-(\bm\pi_B\dot{\bm P}_B\bm 1)^2+2\bm\pi_B\dot{\bm P}_B(\bm Z_B-\bm1\bm\pi_B)\dot{\bm P}_B\bm 1,\\
\end{equation*}
where $\bm1$ is the column vector of 1s of the appropriate dimension, $\dot{\bm P}_B$ is $\bm P_B$ with each $q_m$ replaced by $-q_m$, and $\ddot{\bm P}_B=\bm P_B$.  This ``rule of thumb'' for $\dot{\bm P}_B$ requires some caution:  It must be applied before any simplifications to $\bm P_B$ are made using $q_m=1-p_m$.  Of course, $\bm\pi_B$ is the unique stationary distribution, and $\bm Z_B$ is the fundamental matrix, of $\bm P_B$.

\begin{theorem}\label{means,variances,B}
\begin{equation*}
\mu_B^\circ=\mu_B
\end{equation*}
and
\begin{equation*}
(\sigma_B^\circ)^2=\sigma_B^2.
\end{equation*}
\end{theorem}

\begin{remark}
The proof is as in Ethier and Lee \cite{EL17}.  Let us explain its significance.  $\mu_B^\circ$ and $(\sigma_B^\circ)^2$ are the mean and variance that appear in the SLLN and the CLT.  They are defined in terms of $\bm P_B^\circ$, the augmented one-step transition matrix.  $\mu_B$ and $\sigma_B^2$ are defined analogously in terms of $\bm P_B$, the original one-step transition matrix, using the rule of thumb.
\end{remark}

\subsection{Application to game $C':=\gamma A'+(1-\gamma)B$}\label{SLLN-gameC'}

This case is not much different from the previous one.  Notice that, if game $A'$ is played, the profit to the set of $N$ players is 0, since game $A'$ simply redistributes capital among the players.  So we can use the same augmentation of the state space as before, except that 0 is now a possible value of the profit from the last game played.  In other words, $\Sigma^\circ:=\{0,1\}^N\times\{-1,0,1\}$.  The transition probabilities require some new notation.  Let $\eta^{x,x\pm1,\pm1}$ be the element of $\{0,1\}^N$ representing the players' status after player $x$ plays player $x\pm1$ and wins (1) or loses ($-1$).  Of course player 0 is player $N$ and player $N+1$ is player 1.  E.g., $\eta^{1,2,-1}=(0,1,\eta(3),\ldots,\eta(N))$ (player 1 competes against player 2 and loses, leaving player 1 a loser and player 2 a winner, regardless of their previous status).  Then
\begin{align}\label{PC'circ1}
P_{C'}^\circ((\eta,s),(\eta_x,1))&=\begin{cases}(1-\gamma)N^{-1}p_{m_x(\eta)}&\text{if $\eta(x)=0$,}\\0&\text{if $\eta(x)=1$,}\end{cases}\\
P_{C'}^\circ((\eta,s),(\eta_x,-1))&=\begin{cases}0&\text{if $\eta(x)=0$,}\\(1-\gamma)N^{-1}q_{m_x(\eta)}&\text{if $\eta(x)=1$,}\end{cases}\\
P_{C'}^\circ((\eta,s),(\eta^{x,x-1,-1},0))&=\gamma(4N)^{-1},\\
P_{C'}^\circ((\eta,s),(\eta^{x,x-1,1},0))&=\gamma(4N)^{-1},\\
P_{C'}^\circ((\eta,s),(\eta^{x,x+1,-1},0))&=\gamma(4N)^{-1},\\
P_{C'}^\circ((\eta,s),(\eta^{x,x+1,1},0))&=\gamma(4N)^{-1},
\end{align}
for $x=1,2,\ldots,N$, and
\begin{align}
P_{C'}^\circ((\eta,s),(\eta,1))&=(1-\gamma)N^{-1}\sum_{x:\eta(x)=1}p_{m_x(\eta)},\\  \label{PC'circ8}
P_{C'}^\circ((\eta,s),(\eta,-1))&=(1-\gamma)N^{-1}\sum_{x:\eta(x)=0}q_{m_x(\eta)}.
\end{align}

Of course, we could also define $\bm P_{C'}=\gamma\bm P_{A'}+(1-\gamma)\bm P_B$.  We notice that Theorems~\ref{SLLN/CLT-PBcirc} and \ref{means,variances,B} hold in this framework without change. 

Let $0<p_m<1$ for $m=0,1,2$ or for $m=1,2,3$, so that the Markov chain with one-step transition matrix $\bm P_{C'}^\circ:=\gamma \bm P_{A'}^\circ+(1-\gamma)\bm P_B^\circ$ is ergodic, and let the row vector $\bm\pi_{C'}^\circ$ be its unique stationary distribution.  Define
\begin{align*}
\mu_{(\gamma,1-\gamma)'}^\circ&=\bm\pi_{C'}^\circ\dot{\bm P}_{C'}^\circ\bm1,\\
(\sigma_{(\gamma,1-\gamma)'}^\circ)^2&=\bm\pi_{C'}^\circ\ddot{\bm P}_{C'}^\circ\bm 1-(\bm\pi_{C'}^\circ\dot{\bm P}_{C'}^\circ\bm 1)^2+2\bm\pi_{C'}^\circ\dot{\bm P}_{C'}^\circ(\bm Z_{C'}^\circ-\bm1\bm\pi_{C'}^\circ)\dot{\bm P}_{C'}^\circ\bm 1.
\end{align*}
where $\bm1$ denotes a column vector of $1$s with entries indexed by $\Sigma^\circ$ and $\bm Z_{C'}^\circ:=(\bm I-(\bm P_{C'}^\circ-\bm1\bm\pi_{C'}^\circ))^{-1}$ is the fundamental matrix.  (Notice that $\bm1\bm\pi_{C'}^\circ$ is the square matrix each of whose rows is equal to $\bm\pi_{C'}^\circ$.)  Let $\{X_n^\circ\}_{n\ge0}$ be a time-homogeneous Markov chain in $\Sigma^\circ$ with one-step transition matrix $\bm P_{C'}^\circ$.  For each $n\ge1$, define $\xi_n:=w^\circ(X_{n-1}^\circ,X_n^\circ)$ and $S_n:=\xi_1+\cdots+\xi_n$.  

\begin{theorem}\label{SLLN/CLT-PC'circ}
Under the above assumptions, and with the distribution of $X_0$ arbitrary, 
$$
\lim_{n\to\infty}\frac{S_n}{n}=\mu_{(\gamma,1-\gamma)'}^\circ\;\;\emph{a.s.}  
$$
and, if $(\sigma_{(\gamma,1-\gamma)'}^\circ)^2>0$, then
$$
\frac{S_n-n\mu_{(\gamma,1-\gamma)'}^\circ}{\sqrt{n(\sigma_{(\gamma,1-\gamma)'}^\circ)^2}}\to_d N(0,1)\text{ as }n\to\infty.
$$
\end{theorem}

Let us define
\begin{align*}
\mu_{(\gamma,1-\gamma)'}&:=\bm\pi_{C'}\dot{\bm P}_{C'}\bm1,\\
\sigma_{(\gamma,1-\gamma)'}^2&:=\bm\pi_{C'}\ddot{\bm P}_{C'}\bm 1-(\bm\pi_{C'}\dot{\bm P}_{C'}\bm 1)^2+2\bm\pi_{C'}\dot{\bm P}_{C'}(\bm Z_{C'}-\bm1\bm\pi_{C'})\dot{\bm P}_{C'}\bm 1,
\end{align*}
where $\bm1$ is the column vector of 1s of the appropriate dimension, and since $\dot{\bm P}_{A'}$ can be defined to be $\bm0$, $\dot{\bm P}_{C'}$ is $(1-\gamma)\dot{\bm P}_B$ with each $q_m$ replaced by $-q_m$, and $\ddot{\bm P}_{C'}=(1-\gamma)\bm P_B$.  This ``rule of thumb'' for $\dot{\bm P}_{C'}$ requires some caution:  It must be applied before any simplifications to $\bm P_{C'}$ are made using $q_m=1-p_m$.  Of course, $\bm\pi_{C'}$ is the unique stationary distribution, and $\bm Z_{C'}$ is the fundamental matrix, of $\bm P_{C'}$.  Notice that $\dot{\bm P}_{A'}^\circ=\bm0$, so $\dot{\bm P}_{C'}^\circ=(1-\gamma)\dot{\bm P}_B^\circ$ and $\ddot{\bm P}_{C'}^\circ=(1-\gamma)\ddot{\bm P}_B^\circ$.

\begin{theorem}\label{means,variances,C'}
\begin{equation}\label{means}
\mu_{(\gamma,1-\gamma)'}^\circ=\mu_{(\gamma,1-\gamma)'}
\end{equation}
and
\begin{equation}\label{variances}
(\sigma_{(\gamma,1-\gamma)'}^\circ)^2=\sigma_{(\gamma,1-\gamma)'}^2.
\end{equation}
\end{theorem}

\subsection{Application to game $C':=(A')^r B^s$}\label{SLLN-game[r,s]}

Next we need versions of the SLLN and the CLT suited to game $C':=(A')^r B^s$.  The key result is Theorem \ref{SLLN2}.  

For the same reason as before, the theorem does not apply directly to $\bm P_{A'}$ and $\bm P_B$.  Therefore we again consider the Markov chains in the augmented state space $\Sigma^\circ:=\{0,1\}^N\times\{-1,0,1\}$ with one-step transition matrix $\bm P_{A'}^\circ$ and $\bm P_B^\circ$.  The definitions are as in \eqref{PC'circ1}--\eqref{PC'circ8} with $\gamma=1$ or $\gamma=0$.  With $\bm W^\circ$ as before, the theorem applies.

Fix $r,s\ge1$.  Assume that $\bm P^\circ:=(\bm P_{A'}^\circ)^r(\bm P_B^\circ)^s$, as well as all cyclic permutations of $(\bm P_{A'}^\circ)^r(\bm P_B^\circ)^s$, are ergodic, and let the row vector $\bm\pi^\circ$ be the unique stationary distribution of $\bm P^\circ$.  Let
\begin{equation*}
\mu_{[r,s]'}^\circ:=\frac{1}{r+s}\sum_{v=0}^{s-1}\bm\pi^\circ(\bm P_{A'}^\circ)^r(\bm P_B^\circ)^v\dot{\bm P}_B^\circ\bm1
\end{equation*}
and
\begin{align*}
&(\sigma_{[r,s]'}^\circ)^2\\
&=\frac{1}{r+s}\bigg\{s-\sum_{v=0}^{s-1}(\bm\pi^\circ(\bm P_{A'}^\circ)^r(\bm P_B^\circ)^v\dot{\bm P}_B^\circ\bm1)^2\\
&\qquad\;\;{}+2\bigg[\sum_{0\le u<v\le s-1}\bm\pi^\circ(\bm P_{A'}^\circ)^r(\bm P_B^\circ)^u\dot{\bm P}_B^\circ((\bm P_B^\circ)^{v-u-1}\\
&\qquad\qquad\qquad\qquad\qquad\qquad{}-\bm1\bm\pi^\circ(\bm P_{A'}^\circ)^r(\bm P_B^\circ)^v)\dot{\bm P}_B^\circ\bm1\\
&\qquad\;\;{}+\sum_{u=0}^{s-1}\sum_{v=0}^{s-1}\bm\pi^\circ(\bm P_{A'}^\circ)^r(\bm P_B^\circ)^u\dot{\bm P}_B^\circ(\bm P_B^\circ)^{s-u-1}(\bm Z^\circ-\bm1\bm\pi^\circ)(\bm P_{A'}^\circ)^r(\bm P_B^\circ)^v\dot{\bm P}_B^\circ\bm1\bigg]\bigg\}.
\end{align*}
Let $\{X_n^\circ\}_{n\ge0}$ be a nonhomogeneous Markov chain in $\Sigma^\circ$ with one-step transition matrices $\bm P_{A'}^\circ,\ldots,\bm P_{A'}^\circ$ $(r\text{ times})$, $\bm P_B^\circ,\ldots,\bm P_B^\circ$ $(s\text{ times})$, $\bm P_{A'}^\circ,\ldots,\bm P_{A'}^\circ$ $(r\text{ times})$, $\bm P_B^\circ,\ldots,\bm P_B^\circ$ $(s\text{ times})$, and so on.  For each $n\ge1$, define $\xi_n:=w^\circ(X_{n-1}^\circ,X_n^\circ)$ and $S_n:=\xi_1+\cdots+\xi_n$.

\begin{theorem}\label{SLLN/CLT2}
Under the above assumptions, and with the distribution of $X_0$ arbitrary, 
$$
\frac{S_n}{n}\to\mu_{[r,s]'}^\circ\;\;\emph{a.s.}  
$$
and, if $(\sigma_{[r,s]'}^\circ)^2>0$, then
$$
\frac{S_n-n\mu_{[r,s]'}^\circ}{\sqrt{n(\sigma_{[r,s]'}^\circ)^2}}\to_d N(0,1) \text{ as } n\to\infty.
$$
\end{theorem}

Again there are simpler expressions for this mean and variance.  We define $\mu_{[r,s]'}$ in terms of $\bm\pi$, $\bm P_{A'}$, $\bm P_B$, and $\dot{\bm P}_B$ in the same way that $\mu_{[r,s]'}^\circ$ was defined in terms of $\bm\pi^\circ$, $\bm P_{A'}^\circ$, $\bm P_B^\circ$, and $\dot{\bm P}_B^\circ$.  ($\dot{\bm P}_B$ is defined by the rule of thumb.)  Finally, $\sigma_{[r,s]'}^2$ is defined analogously to $(\sigma_{[r,s]'}^\circ)^2$.

\begin{theorem}\label{means,variances[r,s]}
\begin{equation}\label{means-pattern}
\mu_{[r,s]'}^\circ=\mu_{[r,s]'}
\end{equation}
and 
\begin{equation}\label{variances-pattern}
(\sigma_{[r,s]'}^\circ)^2=\sigma_{[r,s]'}^2.
\end{equation}
\end{theorem} 

\begin{proof}
Eqs.~\eqref{means-pattern} and \eqref{variances-pattern} are proved in the same way as \eqref{means} and \eqref{variances}.
\end{proof}

\section{Numerical computations}\label{numerical}
In this section, we compute various means numerically by using the reduced state space and use computer graphics to visualize the Parrondo region of the Parrondo games of Xie at al.~\cite{XC11}.

\subsection{State-space reduction} 

Let us begin by explaining what we mean by state-space reduction, which is an important method for simplifying our computations.

In general, consider an equivalence relation $\sim$ on a finite set $E$.  By definition, $\sim$ is \textit{reflexive} ($x\sim x$), \textit{symmetric} ($x\sim y$ implies $y\sim x$), and \textit{transitive} ($x\sim y$ and $y\sim z$ imply $x\sim z$).  It is well known that an equivalence relation partitions the set $E$ into \textit{equivalence classes}.  The set of all equivalence classes, called the \textit{quotient set}, will be denoted by $\bar E$.  Let us write $[x]:=\{y\in E: y\sim x\}$ for the equivalence class containing $x$.  Then $\bar E=\{[x]:x\in E\}$.

Now suppose $X_0,X_1,X_2,\ldots$ is a (time-homogeneous) Markov chain in $E$ with transition matrix $\bm P$.  In particular, $P(x,y)=\P(X_{t+1}=y\mid X_t=x)$ for all $x,y\in E$ and $t=0,1,2,\ldots$.  Under what conditions on $\bm P$ is $[X_0],[X_1],[X_2],\ldots$ a Markov chain in the ``reduced'' state space $\bar E$?  A sufficient condition, apparently due to Kemeny and Snell \cite[p.~124]{KS76}, is that $\bm P$ be \textit{lumpable} with respect to $\sim$.  By definition, this means that, for all $x,x',y\in E$,
\begin{equation}\label{lumpability}
x\sim x'\quad\text{implies}\quad \sum_{y'\in[y]}P(x,y')=\sum_{y'\in[y]}P(x',y').
\end{equation}
Moreover, if \eqref{lumpability} holds, then the Markov chain $[X_0],[X_1],[X_2],\ldots$ in $\bar E$ has transition matrix $\bar{\bm P}$ given by
\begin{equation}\label{Pbar}
\bar P([x],[y]):=\sum_{y'\in[y]}P(x,y').
\end{equation}
Notice that \eqref{lumpability} ensures that \eqref{Pbar} is well defined.

For Parrondo games with one-dimensional spatial dependence, the state space, assuming $N\ge3$ players, is
$$
\{\eta=(\eta(1),\eta(2),\ldots,\eta(N)): \eta(x)\in\{0,1\}{\rm\ for\ }x=1,2,\ldots,N\}=\{0,1\}^N,
$$
which has $2^N$ states.  A state $\eta\in\{0,1\}^N$ describes the status of each of the $N$ players, 0 for losers and 1 for winners.  We can also think of $\{0,1\}^N$ as the set of $N$-bit binary representations of the integers $0,1,\ldots,2^N-1$, thereby giving a natural ordering to the vectors in $\{0,1\}^N$.

Ethier and Lee \cite{EL12a} used the following equivalence relation on $\{0,1\}^N$:  $\eta\sim\zeta$ if and only if $\zeta=\eta_\sigma:=(\eta(\sigma(1)),\ldots,\eta(\sigma(N)))$ for a permutation $\sigma$ of $(1,2,\ldots,N)$ belonging to the cyclic group $G$ of order $N$ of the rotations of the players.  If, in addition, $p_1=p_2$, the permutation $\sigma$ can belong to the dihedral group $G$ of order $2N$ of the rotations and reflections of the players.  They verified the lumpability condition, with the result that the size of the state space was reduced by a factor of nearly $N$ (or $2N$ if $p_1=p_2$) for large $N$.  It should be noted that a sufficient condition for the lumpability condition in this setting is that, for every $\eta,\zeta\in\{0,1\}^N$,
\begin{equation*}
P(\eta_\sigma,\zeta_\sigma)=P(\eta,\zeta)\quad\text{for all $\sigma\in G$}
\end{equation*}
or for all $\sigma$ in a subset of $G$ that generates $G$.

To fully justify this, the following lemma is useful.

\begin{lemma}[Ethier and Lee \cite{EL12b}]
Fix $N\ge3$, let $G$ be a subgroup of the symmetric group $S_N$.  Let $\bm P$ be the one-step transition matrix for a Markov chain in $\{0,1\}^N$ with a unique stationary distribution $\bm\pi$.  Assume that
\begin{equation}\label{G-invariance}
P(\eta_\sigma,\zeta_\sigma)=P(\eta,\zeta),\qquad \sigma\in G,\;\eta,\zeta\in\{0,1\}^N.
\end{equation}
Then $\pi(\eta_\sigma)=\pi(\eta)$ for all $\sigma\in G$ and $\eta\in\{0,1\}^N$.

Let us say that $\eta\in\{0,1\}^N$ is equivalent to $\zeta\in\{0,1\}^N$ (written $\eta\sim\zeta$) if there exists $\sigma\in G$ such that $\zeta=\eta_\sigma$, and let us denote the equivalence class containing $\eta$ by $[\eta]$.  Then, in addition, $\bm P$ induces a one-step transition matrix $\bar{\bm P}$ for a Markov chain in the quotient set (i.e., the set of equivalence classes) $\bar\Sigma$ defined by the formula
\begin{equation*}
\bar{P}([\eta],[\zeta]):=\sum_{\zeta'\in[\zeta]}P(\eta,\zeta'),
\end{equation*}
Furthermore, if $\bar{\bm P}$ has a unique stationary distribution $\bar{\bm\pi}$, then the unique stationary distribution $\bm\pi$ is given by $\pi(\eta)=\bar{\pi}([\eta])/|[\eta]|$, where $|[\eta]|$ denotes the cardinality of the equivalence class $[\eta]$.
\end{lemma}

The lemma will apply to $\bm P_{A'}$ and $\bm P_B$ (hence $\bm P_{C'}$) if we can verify \eqref{G-invariance} for $G$ being the cyclic group of rotations or, if $p_1=p_2$, the dihedral group of rotations and reflections.

The practical effect of this is that we can reduce the size of the state space (namely, $2^N$) to what we will call its \textit{effective size}, which is simply the number of equivalence classes.  For example, if $N=3$, there are eight states and four equivalence classes, namely
$$
0=\{000\},\quad
1=\{001,010,100\},\quad
2=\{011,101,110\},\quad
3=\{111\}.
$$
Notice that we label equivalence classes by the number of 1s each element has.  If $N=4$, there are 16 states and six equivalence classes, namely
\begin{align*}
0&=\{0000\},\\
1&=\{0001,0010,0100,1000\},\\
2&=\{0011,0110,1001,1100\},\\
2'&=\{0101,1010\},\\
3&=\{0111,1011,1101,1110\},\\
4&=\{1111\}.
\end{align*}
In these two cases, it does not matter which of the two equivalence relations we use; the result is the same.

The number of equivalence classes with $G$ being the group of cyclic permutations follows the sequence A000031 in the \textit{The On-Line Encyclopedia of Integer Sequences} (Sloan \cite{S19}), described as the number of necklaces with $N$ beads of two colors when turning over is not allowed.  There is an explicit formula in terms of Euler's phi-function.  If $p_1=p_2$, we can reverse the order of the players, and the number of equivalence classes with $G$ being the dihedral group follows the sequence A000029 in the \textit{OEIS}, described as the number of necklaces with $N$ beads of two colors when turning over is allowed.  Again there is an explicit formula.  

\begin{example}
To illustrate our approach in a tractable case, we focus on the case $N=4$, as did Xie et al.~\cite{XC11}.  Here $\{0,1\}^N$ has 16 states, ordered as the 4-bit binary representations of the number 0--15.  First, $\bm P_B$ has the form as in Eq.\ (12) of Xie et al.~\cite{XC11}.  For example, the diagonal entries of $4\bm P_B$ are
\begin{align*}
d_0&:=4 q_0,\\
d_1=d_2=d_4=d_8&:=p_0 + q_0 + q_1 + q_2,\\
d_3=d_6=d_9=d_{12}&:=p_1 + p_2 + q_1 + q_2,\\
d_5=d_{10}&:=2 (p_0 + q_3),\\
d_7=d_{11}=d_{13}=d_{14}&:=p_1 + p_2 + p_3 + q_3,\\
d_{15}&:=4 p_3,
\end{align*}
where $q_m:=1-p_m$ for $m=0,1,2,3$.  

For the equivalence relation above, there are six equivalence classes, namely
$\{0000\}$, $\{0001,0010,0100,1000\}$, $\{0011,0110,1001,1100\}$, $\{0101,1010\}$, $\{0111,\break 1011,1101, 1110\}$, and $\{1111\}$.  Denoting the states by their decimal representations (0--15),  the equivalence classes are $\{0\}$, $\{1,2,4,8\}$, $\{3,6,9,$ $12\}$, $\{5,10\}$, $\{7,11,13,14\}$, and $\{15\}$.  It will be convenient to reorder the states temporarily.  Within each equivalence class, we order elements so that each is a fixed rotation of the preceding one, that is, $\{0000\}$, $\{1000,0100,0010,0001\}$, $\{1100,0110,0011,\break 1001\}$, $\{1010,0101\}$, $\{1110,0111,1011,1101\}$, and $\{1111\}$, or $\{0\}$, $\{8,4,2,1\}$, $\{12,6,3,9\}$, $\{10,5\}$, $\{14,7,11,13\}$, and $\{15\}$.  We now order states in this order: 0, 8, 4, 2, 1, 12, 6, 3, 9, 10, 5, 14, 7, 11, 13, 15, which leads to an alternative form for the transition matrix shown in Figure~\ref{P_B-block}

\begin{figure}[htb]
\begin{center}
$$ 
\arraycolsep=1.1mm
\bm P_B:=\frac{1}{4}\left(
\begin{array}{c|cccc|cccc|cc|cccc|c}
d_0 & p_0 & p_0 & p_0 & p_0 & 0      &  0  & 0   & 0      & 0      & 0      & 0      & 0      & 0      & 0      & 0      \\
\noalign{\smallskip}\hline\noalign{\smallskip}
q_0 & d_8 & 0   & 0   & 0   & p_2    & 0   & 0   & p_1    & p_0    & 0      & 0      & 0      & 0      & 0      & 0      \\
q_0 & 0   & d_4 & 0   & 0   & p_1    & p_2 & 0   & 0      & 0      & p_0    & 0      & 0      & 0      & 0      & 0      \\  
q_0 & 0   & 0   & d_2 & 0   & 0      & p_1 & p_2 & 0      & p_0    & 0      & 0      & 0      & 0      & 0      & 0      \\ 
q_0 & 0   & 0   & 0   & d_1 & 0      & 0   & p_1 & p_2    & 0      & p_0    & 0      & 0      & 0      & 0      & 0      \\ 
\noalign{\smallskip}\hline\noalign{\smallskip}
0   & q_2 & q_1 & 0   & 0   & d_{12} & 0   & 0   & 0      & 0      & 0      & p_2    & 0      & 0      & p_1    & 0      \\ 
0   & 0   & q_2 & q_1 & 0   & 0      & d_6 & 0   & 0      & 0      & 0      & p_1    & p_2    & 0      & 0      & 0      \\ 
0   & 0   & 0   & q_2 & q_1 & 0      & 0   & d_3 & 0      & 0      & 0      & 0      & p_1    & p_2    & 0      & 0      \\ 
0   & q_1 & 0   & 0   & q_2 & 0      & 0   & 0   & d_9    & 0      & 0      & 0      & 0      & p_1    & p_2    & 0      \\ 
\noalign{\smallskip}\hline\noalign{\smallskip}
0   & q_0 & 0   & q_0 & 0   & 0      & 0   & 0   & 0      & d_{10} & 0      & p_3    & 0      & p_3    & 0      & 0      \\ 
0   & 0   & q_0 & 0   & q_0 & 0      & 0   & 0   & 0      & 0      & d_5    & 0      & p_3    & 0      & p_3    & 0      \\ 
\noalign{\smallskip}\hline\noalign{\smallskip}
0   & 0   & 0   & 0   & 0   & q_2    & q_1 & 0   & 0      & q_3    & 0      & d_{14} & 0      & 0      & 0      & p_3    \\
0   & 0   & 0   & 0   & 0   & 0      & q_2 & q_1 & 0      & 0      & q_3    & 0      & d_7    & 0      & 0      & p_3    \\ 
0   & 0   & 0   & 0   & 0   & 0      & 0   & q_2 & q_1    & q_3    & 0      & 0      & 0      & d_{11} & 0      & p_3    \\ 
0   & 0   & 0   & 0   & 0   & q_1    & 0   & 0   & q_2    & 0      & q_3    & 0      & 0      & 0      & d_{13} & p_3    \\ 
\noalign{\smallskip}\hline\noalign{\smallskip}
0   & 0   & 0   & 0   & 0   & 0      & 0   & 0   & 0      & 0      & 0      & q_3    & q_3    & q_3    & q_3    & d_{15}  \end{array}\right).
$$
\caption{\label{P_B-block}The transition matrix $\bm P_B$ for the Markov chain describing game $B$, with states ordered first by equivalence class and then in the order $\{0\}, \{8, 4, 2, 1\}, \{12, 6, 3, 9\}, \{10, 5\}, \{14, 7, 11, 13\}, \{15\}$.}
\end{center}
\end{figure}

The lumpability condition requires that, within each block, row sums be equal.  
That this condition is met can be seen at a glance.  Moreover, we can also see that the sufficient condition \eqref{G-invariance} holds as well.  Because of how we ordered the states, this condition requires that each block be constant along each diagonal parallel to the main diagonal (assuming periodic boundary conditions).  

We conclude that
\begin{footnotesize}
\begin{equation*}
\setlength{\arraycolsep}{0.7mm}
\bar{\bm P}_B=\frac{1}{4}\left(\begin{array}{cccccc}
4 q_0 & 4 p_0 & 0 & 0 & 0 & 0 \\ 
q_0 & p_0 + q_0 + q_1 + q_2 & p_1 + p_2 & p_0 & 0 & 0 \\ 
0 & q_1 + q_2 & p_1 + p_2 + q_1 + q_2 & 0 & p_1 + p_2 & 0 \\ 
0 & 2 q_0 & 0 & 2 (p_0 + q_3) & 2 p_3 & 0 \\ 
0 & 0 & q_1 + q_2 & q_3 & p_1 + p_2 + p_3 + q_3 & p_3 \\ 
0 & 0 & 0 & 0 & 4 q_3 & 4 p_3
\end{array}\right).
\end{equation*}
\end{footnotesize}

We turn next to game $A'$.  Again there are 16 states (namely, the 4-bit binary representations of the integers 0--15) and the transition matrix can be easily evaluated.  To verify the lumpability condition we reorder the states and rewrite the matrix in block form as we did for $\bm P_B$.  See Figure~\ref{P_A'-block}.

\begin{figure}[htb]
\begin{center}
$$
\arraycolsep=2mm
\bm P_{A'}:=\frac{1}{8}\left(
\begin{array}{c|cccc|cccc|cc|cccc|c}
0 & 2 & 2 & 2 & 2 & 0 & 0 & 0 & 0 & 0 & 0 & 0 & 0 & 0 & 0 & 0 \\ 
\noalign{\smallskip}\hline\noalign{\smallskip}
0 & 2 & 1 & 0 & 1 & 1 & 0 & 0 & 1 & 2 & 0 & 0 & 0 & 0 & 0 & 0 \\ 
0 & 1 & 2 & 1 & 0 & 1 & 1 & 0 & 0 & 0 & 2 & 0 & 0 & 0 & 0 & 0 \\ 
0 & 0 & 1 & 2 & 1 & 0 & 1 & 1 & 0 & 2 & 0 & 0 & 0 & 0 & 0 & 0 \\ 
0 & 1 & 0 & 1 & 2 & 0 & 0 & 1 & 1 & 0 & 2 & 0 & 0 & 0 & 0 & 0 \\ 
\noalign{\smallskip}\hline\noalign{\smallskip}
0 & 1 & 1 & 0 & 0 & 2 & 0 & 0 & 0 & 1 & 1 & 1 & 0 & 0 & 1 & 0 \\ 
0 & 0 & 1 & 1 & 0 & 0 & 2 & 0 & 0 & 1 & 1 & 1 & 1 & 0 & 0 & 0 \\ 
0 & 0 & 0 & 1 & 1 & 0 & 0 & 2 & 0 & 1 & 1 & 0 & 1 & 1 & 0 & 0 \\ 
0 & 1 & 0 & 0 & 1 & 0 & 0 & 0 & 2 & 1 & 1 & 0 & 0 & 1 & 1 & 0 \\ 
\noalign{\smallskip}\hline\noalign{\smallskip}
0 & 0 & 0 & 0 & 0 & 1 & 1 & 1 & 1 & 4 & 0 & 0 & 0 & 0 & 0 & 0 \\ 
0 & 0 & 0 & 0 & 0 & 1 & 1 & 1 & 1 & 0 & 4 & 0 & 0 & 0 & 0 & 0 \\ 
\noalign{\smallskip}\hline\noalign{\smallskip}
0 & 0 & 0 & 0 & 0 & 1 & 1 & 0 & 0 & 2 & 0 & 2 & 1 & 0 & 1 & 0 \\ 
0 & 0 & 0 & 0 & 0 & 0 & 1 & 1 & 0 & 0 & 2 & 1 & 2 & 1 & 0 & 0 \\ 
0 & 0 & 0 & 0 & 0 & 0 & 0 & 1 & 1 & 2 & 0 & 0 & 1 & 2 & 1 & 0 \\ 
0 & 0 & 0 & 0 & 0 & 1 & 0 & 0 & 1 & 0 & 2 & 1 & 0 & 1 & 2 & 0 \\ 
\noalign{\smallskip}\hline\noalign{\smallskip}
0 & 0 & 0 & 0 & 0 & 0 & 0 & 0 & 0 & 0 & 0 & 2 & 2 & 2 & 2 & 0\end{array}\right).
$$
\caption{\label{P_A'-block}The transition matrix $\bm P_{A'}$ for the Markov chain describing game $A'$, with states ordered first by equivalence class and then in the order $\{0\}, \{8, 4, 2, 1\}, \{12, 6, 3, 9\}, \{10, 5\}, \{14, 7, 11, 13\}, \{15\}$.}
\end{center}
\end{figure}

Again the condition is clearly met, and we have
\begin{equation*}
\setlength{\arraycolsep}{2mm}
\bar{\bm P}_{A'}=\frac{1}{4}\left(
\begin{array}{cccccc}
0 & 4 & 0 & 0 & 0 & 0 \\
0 & 2 & 1 & 1 & 0 & 0 \\
0 & 1 & 1 & 1 & 1 & 0 \\
0 & 0 & 2 & 2 & 0 & 0 \\
0 & 0 & 1 & 1 & 2 & 0 \\
0 & 0 & 0 & 0 & 4 & 0 \\
\end{array}\right).
\end{equation*}
The lumpability condition \eqref{G-invariance} has been checked for $\bm P_B$ by Ethier and Lee \cite{EL12a}.  For $\bm P_{A'}$, we can verify \eqref{G-invariance} by observing that, if $(\sigma(1),\ldots,\sigma(N))=(2,3,\ldots,N,1)$, then, after some calculations, $P_{A'}(\eta_\sigma,\zeta_\sigma)=P_{A'}(\eta,\zeta)$. If $(\sigma(1),\break\ldots,\sigma(N))=(N,N-1,\ldots,2,1)$, then the same identity holds.
\end{example}

\subsection{Means and variances}\label{means,numerical}

We saw in Theorems \ref{means,variances,C'} and \ref{means,variances[r,s]} that the means and variances that appear in the SLLNs and CLTs of Sections \ref{SLLN-gameB}--\ref{SLLN-game[r,s]} (namely, $\mu_B^\circ$, $\mu_{(\gamma,1-\gamma)'}^\circ$, $\mu_{[r,s]'}^\circ$, $(\sigma_B^\circ)^2$, $(\sigma_{(\gamma,1-\gamma)'}^\circ)^2$, and $(\sigma_{[r,s]'}^\circ)^2$) are equal to the corresponding quantities defined in terms of the original transition matrices (namely, $\mu_B$, $\mu_{(\gamma,1-\gamma)'}$, $\mu_{[r,s]'}$, $\sigma_B^2$, $\sigma_{(\gamma,1-\gamma)'}^2$, and $\sigma_{[r,s]'}^2$).  We claim that the corresponding quantities defined in terms of the reduced transition matrices (namely, $\bar{\mu}_B$, $\bar{\mu}_{(\gamma,1-\gamma)'}$, $\bar{\mu}_{[r,s]'}$, $\bar{\sigma}_B^2$, $\bar{\sigma}_{(\gamma,1-\gamma)'}^2$, and $\bar{\sigma}_{[r,s]'}^2$) are also equal.  First, we define
\begin{align*}
\bar{\mu}_B&:=\bar{\bm\pi}_B\dot{\bar{\bm P}}_B\bm1,\\
\bar{\mu}_{(\gamma,1-\gamma)'}&:=(1-\gamma)\bar{\bm\pi}_{C'}\dot{\bar{\bm P}}_B\bm1,\\
\bar{\mu}_{[r,s]'}&:=\frac{1}{r+s}\sum_{v=0}^{s-1}\bar{\bm\pi}\bar{\bm P}_{A'}^r\bar{\bm P}_B^v\dot{\bar{\bm P}}_B\bm1,\\
\bar{\sigma}_B^2&:=\bar{\bm\pi}_B\ddot{\bar{\bm P}}_B\bm 1-(\bar{\bm\pi}_B\dot{\bar{\bm P}}_B\bm 1)^2+2\bar{\bm\pi}_B\dot{\bar{\bm P}}_B(\bar{\bm Z}_B-\bm1\bar{\bm\pi}_B)\dot{\bar{\bm P}}_B\bm 1,\\
\bar{\sigma}_{(\gamma,1-\gamma)'}^2&:=\bar{\bm\pi}_{C'}\ddot{\bar{\bm P}}_{C'}\bm 1-(\bar{\bm\pi}_{C'}\dot{\bar{\bm P}}_{C'}\bm 1)^2+2\bar{\bm\pi}_{C'}\dot{\bar{\bm P}}_{C'}(\bar{\bm Z}_{C'}-\bm1\bar{\bm\pi}_{C'})\dot{\bar{\bm P}}_{C'}\bm 1,\\
\bar{\sigma}_{[r,s]'}^2
&:=\frac{1}{r+s}\bigg\{s-\sum_{v=0}^{s-1}(\bar{\bm\pi}\bar{\bm P}_{A'}^r\bar{\bm P}_B^v\dot{\bar{\bm P}}_B\bm1)^2\nonumber\\
&\qquad\qquad{}+2\bigg[\sum_{0\le u<v\le s-1}\bar{\bm\pi}\bar{\bm P}_{A'}^r\bar{\bm P}_B^u\dot{\bar{\bm P}}_B(\bar{\bm P}_B^{v-u-1}-\bm1\bar{\bm\pi}\bar{\bm P}_{A'}^r\bar{\bm P}_B^v)\dot{\bar{\bm P}}_B\bm1\nonumber\\
&\qquad\qquad\qquad{}+\sum_{u=0}^{s-1}\sum_{v=0}^{s-1}\bar{\bm\pi}\bar{\bm P}_{A'}^r\bar{\bm P}_B^u\dot{\bar{\bm P}}_B\bar{\bm P}_B^{s-u-1}(\bar{\bm Z}-\bm1\bar{\bm\pi})\bar{\bm P}_{A'}^r\bar{\bm P}_B^v\dot{\bar{\bm P}}_B\bm1\bigg]\bigg\}.
\end{align*}

\begin{theorem}
$$
\mu_B=\bar{\mu}_B,\qquad \mu_{(\gamma,1-\gamma)'}=\bar{\mu}_{(\gamma,1-\gamma)'},\qquad \mu_{[r,s]'}=\bar{\mu}_{[r,s]'}
$$
and
$$
\sigma_B^2=\bar{\sigma}_B^2,\qquad \sigma_{(\gamma,1-\gamma)'}^2=\bar{\sigma}_{(\gamma,1-\gamma)'}^2,\qquad \sigma_{[r,s]'}^2=\bar{\sigma}_{[r,s]'}^2.
$$
\end{theorem}

\begin{proof}
A result of Ethier and Lee \cite{EL12b} implies that, if $\bm Q$ is a $G$-invariant square (not necessarily stochastic) matrix (i.e., $Q(\eta_\sigma,\zeta_\sigma)=Q(\eta,\zeta)$ for all $\eta,\zeta\in\{0,1\}^N$ and all $\sigma\in G$), then 
$$
\bm\pi\bm Q\bm1=\bar{\bm\pi}\bar{\bm Q}\bm1.
$$
Repeated application of this identity gives the desired conclusions.
\end{proof}

The formulas for the means with bars are computable for $3\le N\le 18$, at least.  We give partial results for Toral's \cite{T01} choice of the parameter vector $(p_0,p_1,p_2,p_3)$ in Table \ref{ex2}.  The formulas for the variances with bars are perhaps computable for $3\le N\le 12$, but we do not include them here.  

\subsection{Computer graphics}
Ethier and Lee \cite{EL15} sketched, for games $A$, $B$, and $C:=\frac12 A+\frac12 B$, the Parrondo and anti-Parrondo regions when $3\le N\le 9$.  They assumed that $p_1=p_2$ and relabeled $p_3$ as $p_2$.  In other words, their parameter vector was of the form $(p_0,p_1,p_1,p_2)$.  (The reason for this simplification is that a three-dimensional figure is easier to visualize than a four-dimensional figure.)   The figures for games $A'$, $B$, and $C':=\frac12 A'+\frac12 B$ are distinctively different from those for games $A$, $B$, and $C$.  In both cases, the general shape of the Parrondo and anti-Parrondo regions does not change much, once $N\ge5$.  We illustrate in the case $r=1$ and $s=2$ in Figure~\ref{region_A'BB}.

\begin{table}[htb]
\caption{\label{ex2}Mean profit per turn at equilibrium in the games of Toral \cite{T01} and Xie et al.~\cite{XC11}, assuming $(p_0,p_1,p_2,p_3)=(1,4/25,4/25,7/10)$.  Results are given to six significant digits.  The entries corresponding to $N=\infty$ are limits as $N\to\infty$ (see Theorem~\ref{periodic-limit}).}
\catcode`@=\active \def@{\hphantom{0}}
\catcode`#=\active \def#{\hphantom{$-$}}
\tabcolsep=.03cm
\begin{small}
\begin{center}
\begin{tabular}{ccccccc}
\noalign{\smallskip}
\multicolumn{7}{c}{mean profit per turn, Toral's games}\\
\noalign{\smallskip}
\hline
\noalign{\smallskip}
$N$@ &$B$ & $\frac12(A+B)$ & $AB$ & $ABB$ & $AAB$ & $AABB$ \\
\noalign{\smallskip}
\hline
\noalign{\smallskip}
@3@ & $-0.0909091@@$ & $-0.0183774@$ & $-0.00695879$ & $-0.0274821@$ & #0.000672486 & $-0.0148718$@ \\
@6@ & $-0.0189247@@$ & #0.00463310 & #0.00497503 & #0.00590528 & #0.00325099@  & #0.00498178 \\
@9@ & $-0.00189233@$ & #0.00479036 & #0.00493507 & #0.00598135 & #0.00327802@  & #0.00493728 \\
12@ & $-0.000676916$ & #0.00479089 & #0.00490464 & #0.00586697 & #0.00328800@  & #0.00490531 \\
15@ & $-0.000586184$ & #0.00479089 & #0.00488431 & #0.00579891 & #0.00329249@  & #0.00488449 \\
18@ & $-0.000579652$ & #0.00479089 & #0.00486999 & #0.00575438 & #0.00329483@  & #0.00487001 \\
\noalign{\smallskip}
$\infty$@ &          & #0.00479089 & #0.00479089 & #0.00554084 & #0.00329853@  & #0.00479089 \\
\noalign{\smallskip}
\hline
\noalign{\bigskip}
\multicolumn{7}{c}{mean profit per turn, Xie at al.'s games}\\
\noalign{\smallskip}
\hline
\noalign{\smallskip}
$N$@ & $B$ & $\frac12(A'+B)$ & $A'B$ & $A'BB$ & $A'A'B$ & $A'A'BB$ \\
\noalign{\smallskip}
\hline
\noalign{\smallskip}
@3@ & $-0.0909091@@$ & $-0.0766158@$& $-0.105479@@$ & $-0.102038@@$ & $-0.0724638@$ & $-0.0773252$@\\
@6@ & $-0.0189247@@$ & #0.00671656 & #0.00640351 & #0.00955597 & #0.00363075 & #0.00745377 \\
@9@ & $-0.00189233@$ & #0.00678314 & #0.00676079 & #0.00887095 & #0.00402382 & #0.00705972 \\
12@ & $-0.000676916$ & #0.00678336 & #0.00682799 & #0.00860524 & #0.00419181 & #0.00695667 \\
15@ & $-0.000586184$ & #0.00678336 & #0.00684381 & #0.00845891 & #0.00427852 & #0.00691300 \\
18@ & $-0.000579652$ & #0.00678336 & #0.00684607 & #0.00836539 & #0.00433011 & #0.00688859 \\
\noalign{\smallskip}
$\infty$@ &          &#0.00678336& #0.00678336 & #0.00792947 & #0.00451510 &   #0.00678336 \\
\noalign{\smallskip}
\hline
\end{tabular}
\end{center}
\end{small}
\end{table}

\section{Convergence of means}

Computations suggest that $\mu_{(\gamma,1-\gamma)'}^N$ and $\mu_{[r,s]'}^N$ converge as $N\to\infty$, regardless of the parameters $p_0,p_1,p_2,p_3\in(0,1)$.  We cannot prove this, but we can give sufficient conditions on the parameters for this convergence to hold.  These are
\begin{align}\label{erg-basic}
&\max\bigg[ \bigg|\frac{\gamma}{2}+(1-\gamma)(p_0-p_1)\bigg|,\bigg|\frac{\gamma}{2}+(1-\gamma)(p_2-p_3)\bigg|\bigg]\nonumber\\ 
&\qquad + \max\bigg[\bigg|\frac{\gamma}{2}+(1-\gamma)(p_0-p_2)\bigg|,\bigg|\frac{\gamma}{2}+(1-\gamma)(p_1-p_3)\bigg|\bigg]<1.
\end{align}

\begin{theorem}\label{periodic-limit}
Fix integers $r,s\ge1$ and put $\gamma:=r/(r+s)$.  If \eqref{erg-basic} holds,  then $\lim_{N\to\infty}\mu_{(\gamma,1-\gamma)'}^N$ exists, and $\lim_{N\to\infty}\mu_{[r,s]'}^N=\lim_{N\to\infty}\mu_{(\gamma,1-\gamma)'}^N$.
\end{theorem}

The volume of the subset of the parameter space $[0,1]^4$ for which \eqref{erg-basic} holds with $\gamma=1/2$ is, by \textit{Mathematica}, 5/6.  If we assume that $p_1=p_2$, then the volume of the subset of the parameter space $[0,1]^3$ for which \eqref{erg-basic} holds is, by \textit{Mathematica}, 3/4.  In fact, we plot the three-dimensional volume as a function of $\gamma$ in Figure \ref{vol-plot}.

\begin{figure}[htb]
    \begin{center}
	\includegraphics[width=4in]{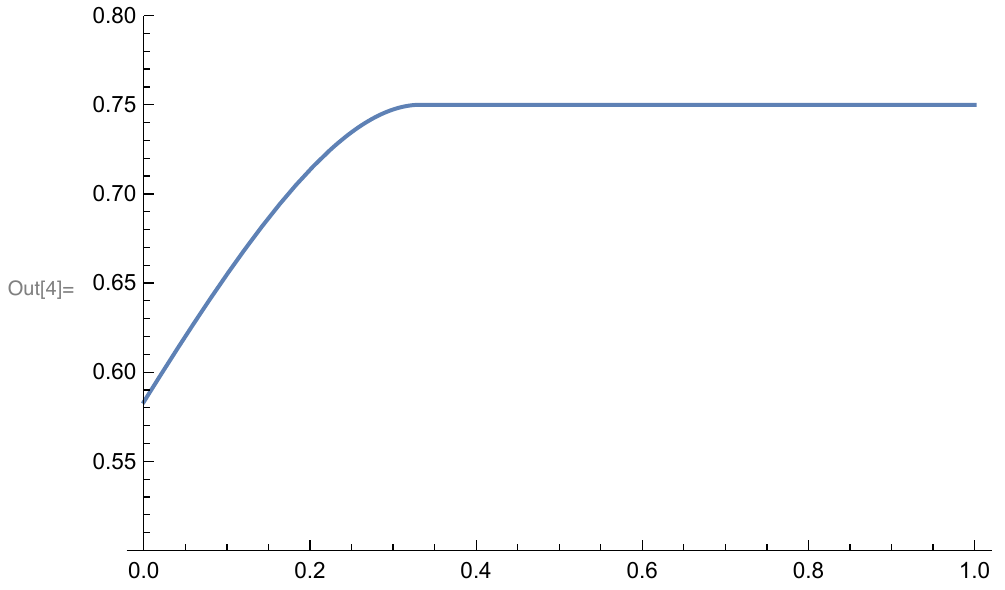}
	\caption{Assuming $p_1=p_2$, the three-dimensional volume of the
		subset of the parameter space for which \eqref{erg-basic} holds
		is plotted as a function of $\gamma$.}\label{vol-plot}
    \end{center}
\end{figure}

Notice that the volume is $3/4$ if and only if $\gamma \ge 1/3$.

\begin{figure}[htb]
\newcommand{\pictwidth}{0.48\textwidth}
\centering
\includegraphics[width = \pictwidth]{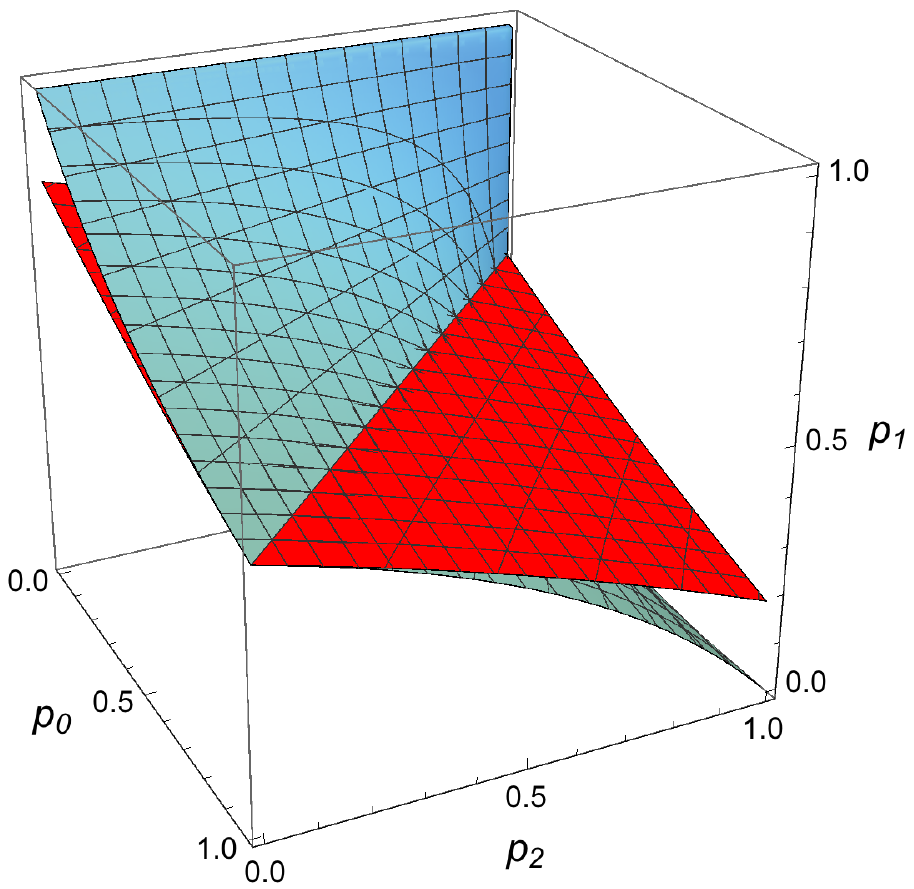}\quad
\includegraphics[width = \pictwidth]{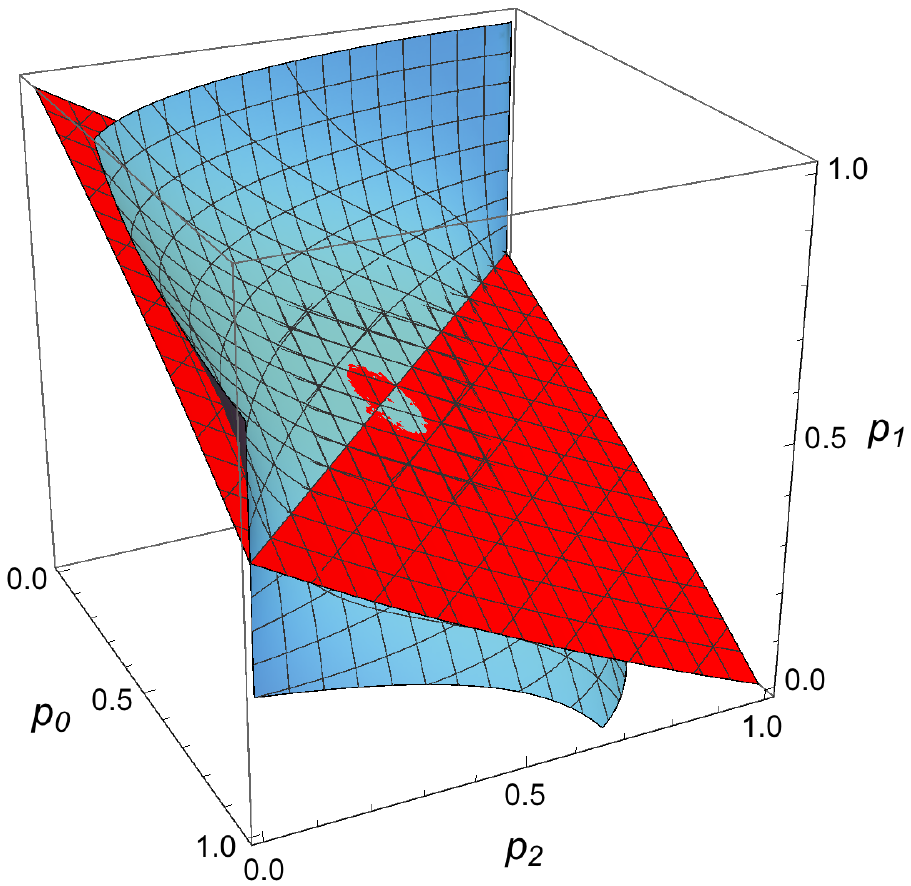}\\
\footnotesize $N=3$ \hspace{1.7in} $N=4$  \\
\bigskip
\includegraphics[width = \pictwidth]{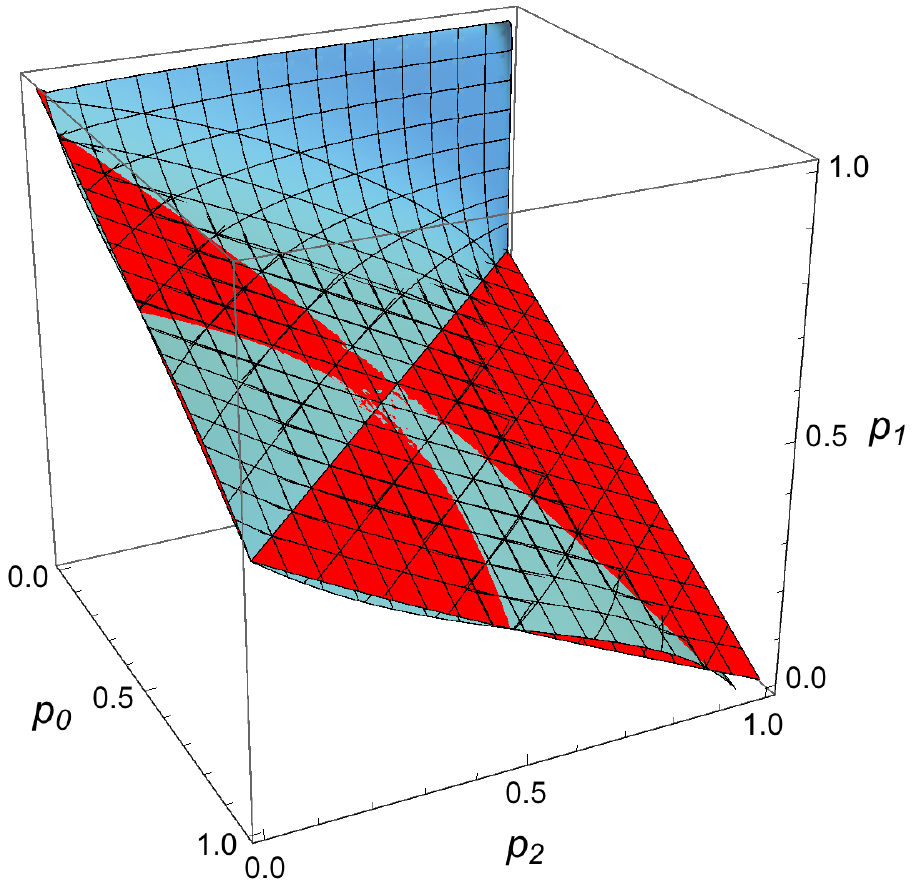}\quad
\includegraphics[width = \pictwidth]{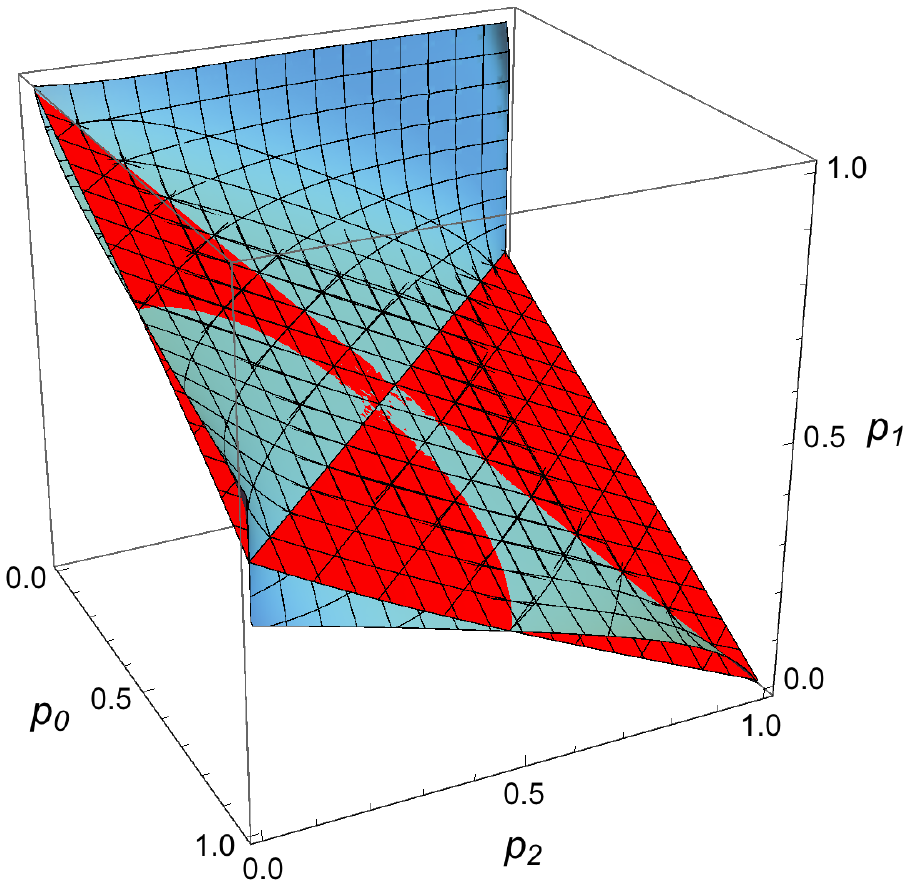}\\
\footnotesize $N=5$  \hspace{1.7in} $N=6$  \\
\caption[The Parrondo region for $C':=A'BB$.]{\label{region_A'BB}For $3\le N\le 6$, the blue surface is the surface $\mu_B=0$, and the red surface is the surface $\mu_{[1,2]'}=0$, in the $(p_0,p_2,p_1)$ unit cube.  The Parrondo region is the region on or below the blue surface and above the red surface, while the anti-Parrondo region is the region on or above the blue surface and below the red surface.  Here $(p_0,p_1,p_1,p_3)$ is relabeled as $(p_0,p_1,p_1,p_2)$.}
\end{figure}

\end{document}